\documentclass[runningheads,a4paper]{llncs}
\usepackage{verbatim}
\usepackage{amssymb}
\usepackage{graphicx}
\usepackage{amsmath}

\begin{document}

\mainmatter  

\title{A Variational Approach to Sparse Model Error Estimation in Cardiac Electrophysiological Imaging}

\titlerunning{Variational sparse error estimation}

%
%
\author{Sandesh Ghimire\inst{1}%
\and John L Sapp\inst{2}\and Milan Horacek\inst{2}\and Linwei Wang\inst{1}
}
%
\authorrunning{S.Ghimire et. al}

\institute{ Rochester Institute of Technology, Rochester, NY 14623, USA \\
\email{sg9872@rit.edu}\\
 \and Dalhouse University, Halifax, NS, Canada\\
}

%
%

\maketitle

\begin{abstract}
Noninvasive reconstruction of cardiac electrical activity 
from surface electrocardiograms (ECG) 
involves solving an ill-posed inverse problem. 
Cardiac electrophysiological (EP) models 
have been used as important \emph{a priori} knowledge 
to constrain this inverse problem. 
However, the reconstruction  
suffer from inaccuracy and uncertainty of the prior model itself
which could be mitigated by estimating \emph{a priori} model error.
Unfortunately, 
due to the need to handle an additional large number of unknowns 
in a problem that already suffers from ill-posedness, 
model error estimation remains an unresolved challenge. 
In this paper, 
we address this issue 
by modeling and estimating the \emph{a priori} model error 
in a low dimensional space 
using a novel sparse prior 
based on the variational approximation of L0 norm.
This prior is used in a 
posterior regularized Bayesian formulation 
to quantify the error in \emph{a priori} EP model 
during  
the reconstruction of transmural action potential 
from ECG data. 
Through synthetic and real-data experiments, 
we demonstrate the ability of the presented method 
to timely capture \emph{a priori} model error 
and thus to improve reconstruction accuracy 
compared to approaches without 
model error correction. 
\end{abstract}

\section{Introduction}
Noninvasive electrophysiological (EP) imaging aims at a mathematical and computational reconstruction of cardiac sources from high density electrocardiogram (ECG) signals. It requires solving an inverse problem with severe ill-posedness, especially when cardiac sources are solved transmurally throughout the 3D myocardium. To overcome this challenge, an important approach is to incorporate a constraining model encoding \emph{a priori} physiological knowledge about the electrical activity inside the heart. Examples of such models include simple step jump functions \cite{pullan01} 
to describe the activation of action potential,  parameterized curve models to describe the spatiotemporal wavefront evolution \cite{ghodrati06}, and biophysical EP models to describe the dynamics of action potential via differential equations \cite{nielsen13,wang10}. While the use of such \emph{a priori} models is effective in regularizing the ill-posed problem, often the inaccuracy and uncertainty of the model itself creates errors and uncertainty in the solution \cite{erem14}.
This issue of model inaccuracy and the resulting solution uncertainty has been studied with convex relaxation of the original problem of EP imaging \cite{erem14}. Similarly, uncertainty in the inverse solution of EP imaging due to model error, and its consideration to improve the clinical interpretation was explored in \cite{xu14}. 
While these works have highlighted the importance of model error in EP imaging, addressing these errors when the prior knowledge is in the form of a complex physiological model remains a challenge. 

Outside EP imaging, 
the issue of error and uncertainty in model prediction has been an active area of research 
in weather and climate forecasting. 
One common approach is to 
model the error as a zero-mean Gaussian distribution 
with an unknown covariance to be inferred from data. 
To address high dimensionality, 
the covariance matrix is usually parameterized by a small set of parameters, 
such as via a linear combination of basis matrices 
weighted by parameters to be learned from data \cite{dee95}. 
Another common approach is to model the error, 
for example as a function of error in model parameters \cite{onatski03}. 
In either case, 
the key is to exploit the prior knowledge 
about the nature of the error to reduce dimension.

In this paper, 
we exploit the low-dimensional nature 
of cardiac wavefront propagation 
to formulate a sparse model 
for the prediction error 
made by a \emph{a priori} EP model.  
Wavefront, 
which can be thought as the spatial gradient of action potential, 
is always localized in a small region at a time. 
This provides the motivation to 
focus on the model inaccuracy in the spatial gradient 
(\textit{i.e.}, error of the predicted wavefront), 
which will reduce the dimension of the unknown error 
while preserving the accuracy of the low-dimensional representation of the inverse solution. 
To do so, 
we present a posterior regularized Bayesian approach 
to the reconstruction of transmural action potential from ECG data, 
in which a generalized Gaussian distribution 
is formulated to model the sparse error 
in the gradient domain. 
A variational lower bound of the generalized Gaussian distribution 
is derived, 
based on which we
quantify the model error while
simultaneously
reconstructing the action potential.  
As we will show mathematically, 
solution of the inference problem will amount to 
iteratively estimating the prior covariance matrix 
of the predicted action potential, 
such that the condition of sparsity in error is satisfied. 
In both synthetic and real-data experiments, 
we demonstrate the ability of the presented method 
to timely capture the prediction error in the \emph{a priori} EP model, 
and thereby to deliver more accurate reconstructions 
of action potential 
compared to approaches without model error correction. 

\section{Bayesian Formulation with Error Modeling}

The relation between ECG 
and transmural action potential 
can be described by a quasi-static approximation 
of the Maxwell's equations for an electromagnetic field \cite{wang10}. 
Solving these equations numerically 
on a subject-specific heart-torso model, 
a linear \emph{forward} model can be obtained as
$
\textbf{y}_k=\textbf{H}\textbf{u}_k$,
where $\textbf{H}$ is called the forward matrix, 
and $ \textbf{y}_k$ and $\textbf{u}_k$ 
respectively denote a vector of ECG and action potential 
at any time instant $k$. 
The inverse problem 
constitutes the estimation of $\textbf{u}_k$ 
given data $ \textbf{y}_k$ for all time $k$.

One of the important constraints used to regularize 
the ill-posed inverse problem of ECG is the 
cardiac EP model describing the spatiotemporal propagation of 
the action potential inside the heart \cite{nielsen13,wang10}. 
Assuming that the temporal evolution of the action potential 
follows a markov model, 
the action potential at every time instant $k$ 
can be predicted from 
its value at the previous time instant 
using the \emph{a priori} EP model. 
Without the loss of generality, 
the two-variable Aliev-Panfilov model \cite{panfilov96} 
is adopted in this paper 
for this purpose.  We drop subscript k for clarity  for the rest of the paper.
 
One common way to calculate the posterior distribution of $\textbf{u}$ is to apply Bayes theorem as
$p(\textbf{u}|\textbf{u}_{prev},\textbf{y})={p(\textbf{y}|\textbf{u})p(\textbf{u}|\textbf{u}_{prev})}/{p(\textbf{y})}$, where $\textbf{u}_{prev}$ is the estimated action potential at previous time instant. In addition, we are also interested in calculating posterior distribution of $\textbf{u}$ such that the value of $\textbf{u}$ is constrained within its physiological range, \emph{i.e.,}$-90 mV \le \textbf{u} \le 20 mV$. This is achieved by finding a distribution $r(\textbf{u})$ which minimizes the Kullback-Leibler(KL) divergence $\textrm{KL}(r(\textbf{u})||p(\textbf{u}|\textbf{u}_{prev},\textbf{y}))$ from $r(\textbf{u})$ to $p(\textbf{u}|\textbf{u}_{prev},\textbf{y})$. This is equivalent to finding an optimum solution of following convex problem \cite{zhu14}:
\begin{align}\label{eq6}
\nonumber
p(\textbf{u}|\textbf{y},\textbf{u}_{prev})&=\underset {r(u)}{\min} \hspace{0.3cm} \textrm{KL}(r(\textbf{u})||p(\textbf{u}|\textbf{u}_{prev}))-\int_{\textbf{u}}{r(\textbf{u})\log p(\textbf{y}|\textbf{u})d\textbf{u}}\\
& s.t. \hspace{0.3cm} r(\textbf{u})\in \mathcal{P}_{post}
\end{align}
where $\mathcal{P}_{post}$ denotes the subspace of Gaussian distributions whose mean lie in the range [-90mV,20mV]. 
\subsubsection{Likelihood Function:}
The forward model is incorporated in the likelihood function 
given by
$
p(\textbf{y}|\textbf{u}, \beta )=\mathcal{N}(\textbf{y}|\textbf{Hu},\beta^{-1}\textbf{I})
$
where $\beta^{-1}$ denotes the unknown variance of data error and the inverse variance, $\beta$, follows a Gamma distribution. 
\subsubsection{Prior Distribution with Error Modeling:}
The prediction from 
the \emph{a priori} EP model 
is incorporated into a prior distribution 
on $\textbf{u}$. 
 \begin{eqnarray}\label{eq11}
 \textbf{u}=\textbf{u}_{pd} + \tilde{\boldsymbol{\eta}} =&&\bar{\textbf{u}}_{pd}+\mathcal{N}(0,\textbf{C}_{pd})+\tilde{\boldsymbol{\eta}}\\
 =&&\bar{\textbf{u}}_{pd}+\boldsymbol{\eta}  
 \end{eqnarray}
where $\textbf{u}_{pd}=f(\textbf{u}_{prev})$ 
represents the distribution predicted by EP model, $f(\cdot)$ 
represent the corresponding function and $\tilde{\boldsymbol{\eta}}$ represents model error .   
Because $f(\cdot)$ is nonlinear, 
$\textbf{u}_{pd}$ is calculated by 
taking a set of samples from 
the estimated posterior distribution of $\textbf{u}_{prev}$, 
passing them through the EP model, 
and fitting a Gaussian distribution to the output samples. 
In most existing works that use such a model 
to constrain the reconstruction of $\textbf{u}$, 
$\tilde{\boldsymbol{\eta}}$ is assumed to be 
a zero-mean Gaussian error with 
a known covariance 
that is experimentally adjusted 
for sub-optimal reconstruction accuracy \cite{wang10}. 
In our case, we absorb all the uncertainties into error vector $\boldsymbol{\eta}=\mathcal{N}(0,\textbf{C}_{pd})+\tilde{\boldsymbol{\eta}}$. To directly estimate $\boldsymbol{\eta}$, 
however, 
will translate to the estimation of 
a high-dimensional covariance matrix 
that is infeasible given the limited ECG data.
We address this problem by 
modeling the prediction error $\boldsymbol{\eta}$ 
to be sparse in the spatial gradient domain, 
utilizing the knowledge about the spatial sparsity 
of the action potential wavefront.  
Let $\textbf{D}$ denote a spatial gradient operator 
on the cardiac mesh, 
we denote 
the error in the gradient domain as  
$\textbf{x}=\textbf{D}\boldsymbol{\eta}$. 

In deterministic cases in the field of compressed sensing \cite{chartrand08}, it has been shown that a 
Lp-norm constraint 
with $0<p<1$ can be used 
to promote sparsity, 
where the sparsity increases as p decreases towards 0. 
In order to use the same concept in a probabilistic setting, 
we propose a generalized Gaussian distribution 
as a prior distribution over variable $\textbf{x}$,  
with the value of $p\approx0$ 
to approximate the effect of a L0 norm constraint: 
\begin{eqnarray}
\label{eqGG}
p_{\textbf{x}}(\textbf{x}|\alpha)&=\prod_i\frac{C}{\alpha}\exp(\frac{-\lvert x_i\rvert^p}{\alpha^p})
\end{eqnarray}
where $x_i, i=1,\cdots,N$ are 
independent elements in the vector $\textbf{x}$, $\alpha$ is the parameter of distribution and $C$ is normalization constant. 
The prior distribution  of $\textbf{u}$ can then be obtained by replacing  $\textbf{x}=\textbf{D}\boldsymbol{\eta}=\textbf{D}(\textbf{u}-
\bar{\textbf{u}}_{pd})$ in eq.(\ref{eqGG}):  
\begin{eqnarray}\label{eq12}
p(\textbf{u}|\bar{\textbf{u}}_{pd})\propto{p_{\textbf{x}}(\textbf{D}(\textbf{u}-
\bar{\textbf{u}}_{pd}))}
\end{eqnarray}

\section{Posterior Regularized Bayes for EP Imaging}
\subsubsection{Variational Lower Bound:} 
The posterior distribution in equation (\ref{eq6}) is analytically intractable 
given the 
generalized Gaussian prior as defined in equation (\ref{eq12}) due to presence of absolute value in the exponent.
Below we present a solution strategy 
based on its variational lower bound. 
\begin{lemma}
$ \exp(\frac{-\lvert x \rvert^p}{\alpha^p})\geq \exp\left(-\frac{x^2}{2\tau}-\frac{2-p}{2}(\frac{\alpha^2}{p\tau})^{\frac{p}{p-2}}\right), \hspace{0.01cm}\forall\tau>0, x \in \bbbr, \alpha>0$ \cite{palmer06}
\end{lemma}

\begin{theorem}
Let $\textbf{x}=(x_1,x_2,...x_N)$ be a vector with independent components each following a generalized normal distribution with same parameters $\alpha, p$. 
then, 
$$
p(\textbf{x}|\alpha)\geq\frac{C'}{\alpha^N}\exp(-\frac{\textbf{x}^T\textbf{Ax}}{2})\exp\left(-\frac{2-p}{2}(\frac{\alpha^2}{p})^{\frac{p}{p-2}}\sum_i\lambda_i^{\frac{p}{p-2}}\right)$$
\end{theorem}
\begin{proof}
Using Lemma 1 in prior distribution as defined in eq.(\ref{eqGG}), we have
\begin{eqnarray}
p(\textbf{x}|\alpha)&\geq\frac{C'}{\alpha^N}\exp(-\sum_i\frac{x_i^2}{2\tau_i})\exp\left(-\sum_i\frac{2-p}{2}(\frac{\alpha^2}{p\tau_i})^{\frac{p}{p-2}}\right)\\\label{eq2}
&=\frac{C'}{\alpha^N}\exp(-\frac{\textbf{x}^T\textbf{Ax}}{2})\exp\left(-\frac{2-p}{2}(\frac{\alpha^2}{p})^{\frac{p}{p-2}}\sum_i\lambda_i^{\frac{p}{p-2}}\right)
\end{eqnarray}
where $\textbf{A}=diag(\boldsymbol{\lambda})$, $C'=C^N$, and $\lambda=\tau^{-1}$ is used in eq.(\ref{eq2})
\qed
\end{proof}
Theorem 1 gives us a variational lower bound 
for the generalized Gaussian distribution 
defined in equation (\ref{eqGG}). 
Applying it to equation (\ref{eq12}), 
we obtain a variational lower bound 
for the prior distribution of 
$p(\textbf{u}|\textbf{u}_{pd})$ 
as follows: 
\begin{eqnarray}\label{eq3}
p(\textbf{u}|\bar{\textbf{u}}_{pd}) &\ge & 
q(\textbf{u}|\bar{\textbf{u}}_{pd}, \boldsymbol{\lambda},\alpha)\\ \nonumber
\small
&\propto& \frac{C'}{\alpha^N}\exp(\frac{-(\textbf{u}-\bar{\textbf{u}}_{pd})^T\textbf{D}^T\textbf{AD}(\textbf{u}-\bar{\textbf{u}}_{pd})}{2}-\frac{2-p}{2}(\frac{\alpha^2}{p})^{\frac{p}{p-2}}\sum_i\lambda_i^{\frac{p}{p-2}})
\normalsize
\end{eqnarray}

\subsubsection{ Inference:} 
The posterior distribution in equation (\ref{eq6}) 
can now be solved by 
maximizing the variational parameters of prior distribution as defined in eq.(\ref{eq3}).
It is achieved by following two steps: 
1) estimate 
parameters 
$\boldsymbol{\lambda},\alpha$ and $\beta$ 
by Expectation Maximization (EM), 
and 
2) obtain the posterior distribution of $\textbf{u}$ 
by posterior regularized Bayes \cite{zhu14}.

Parameters $\boldsymbol{\lambda},\alpha$ and $\beta$  
are estimated using EM
by marginalizing over hidden variable $\textbf{u}$ where the two steps : E step and M step are iterated until convergence. In the E-step,
we take expectation and obtain a function $\mathcal{L}$ as
\begin{eqnarray}
\mathcal{L}(\alpha,\beta, \boldsymbol{\lambda})=E_{q(\textbf{u})}[log \left( p(\textbf{u,y}|\alpha, \boldsymbol{\lambda},\beta)\right)+log (\beta)]
\end{eqnarray}
where $E_{q(\textbf{u})}[\hspace{0.1cm}]$ denotes the expectation with respect to the posterior distribution 
of $\textbf{u}$ with fixed parameters 
{$\boldsymbol{\alpha}_{old}, \boldsymbol{\beta}_{old}, \boldsymbol{\lambda}_{old}$}
obtained in previous iteration.
The M-step consists of the following maximization: 
\begin{eqnarray}\label{eq8}
\alpha_{new},\beta_{new},\boldsymbol{\lambda}_{new} &=\underset {\alpha,\beta, \boldsymbol{\lambda}}{argmax} \hspace{0.5 cm} \mathcal{L}(\alpha,\beta, \boldsymbol{\lambda})
\end{eqnarray}

It is noteworthy that for fixed values of parameters,  
the prior distribution 
as defined in equation (\ref{eq3}) 
is Gaussian with the covariance matrix $\textbf{(D}^T \textbf{AD)}^{-1}$ and $\textbf{A}=diag(\boldsymbol{\lambda})$.  
This gives an interesting interpretation of the presented EM procedure as an iterative process to estimate the prior covariance matrix. 
In specific, 
if we focus on the wavefront $\textbf{Du}$, 
the diagonal element of the matrix $\textbf{A}^{-1}$ 
describes the variance of 
each individual element in $\textbf{Du}$. 
We will 
show in section 4 how this estimated 
variance is related to the \emph{a priori} model error.

Once the values of parameters $\hat{\alpha},\hat{\beta}, \hat{\boldsymbol{\lambda}}$  are obtained by the EM procedure, $\hat{\boldsymbol{\lambda}}, \hat{\alpha}$ are plugged in eq.(\ref{eq3}) to obtain prior distribution $q(\textbf{u}|\bar{\textbf{u}}_{pd},\hat{\boldsymbol{\lambda}},\hat{\alpha},)$ which is then used in eq.(\ref{eq6}) to obtain posterior distribution . Since a Gaussian distribution is specified by mean and covariance, the minimization of (\ref{eq6}) 
is done over a space of mean vectors and covariance matrices satisfying the given constraints.

\section{Results}
While we described a general algorithm for model error estimation, we test it for the detection of errors in \emph{a priori} EP model due to the presence of myocardial scar tissue. Using synthetic and real data experiments, we compare the performance of 
the proposed method to  
that without model error estimation \cite{wang10}
\subsubsection{Synthetic Experiments:}
\begin{figure}[tb!]
\centering
\includegraphics[width=0.85\textwidth]{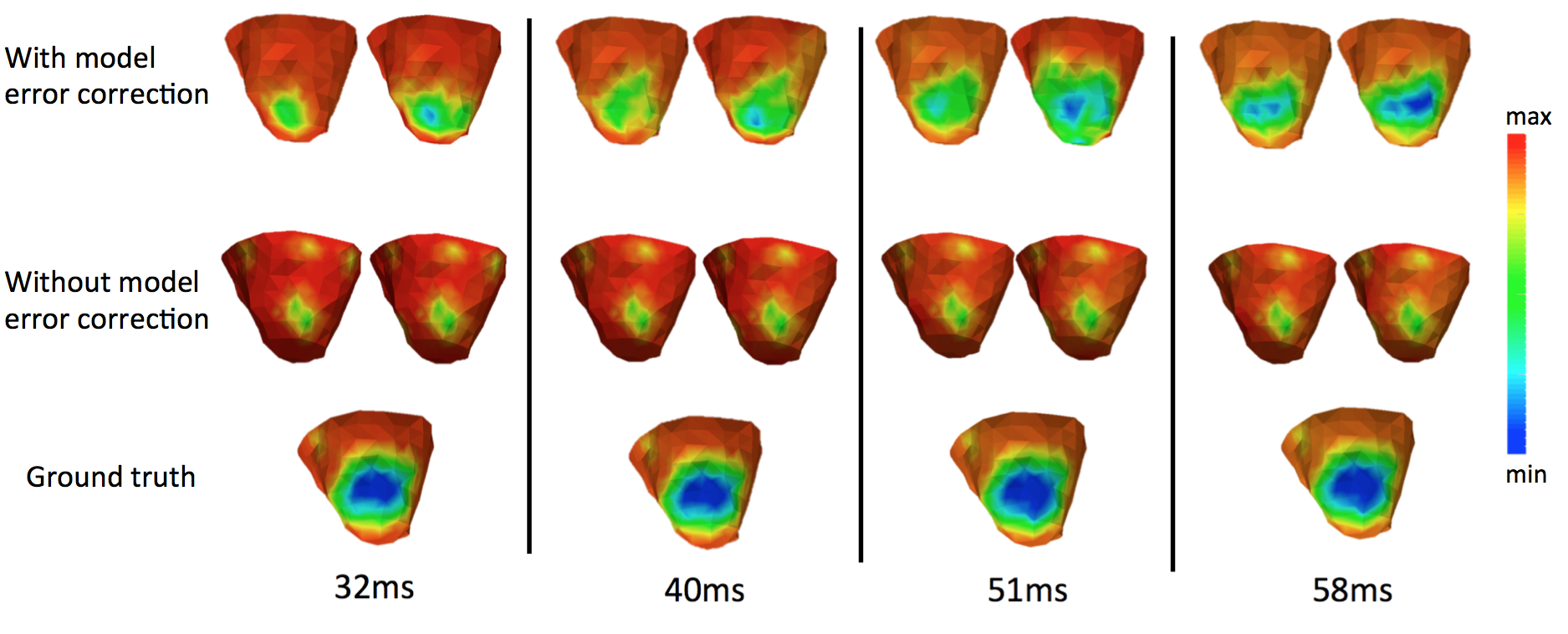}
\caption{Snapshots of the action potential at different time instants.  For each method at each time instant, the figure on the left shows the prediction from the \emph{a priori} EP model; the figure on the right shows the estimated result.}
\label{Modelcorrection}
\end{figure}
\begin{figure}[tb!]
\centering
\includegraphics[width=.65\textwidth]{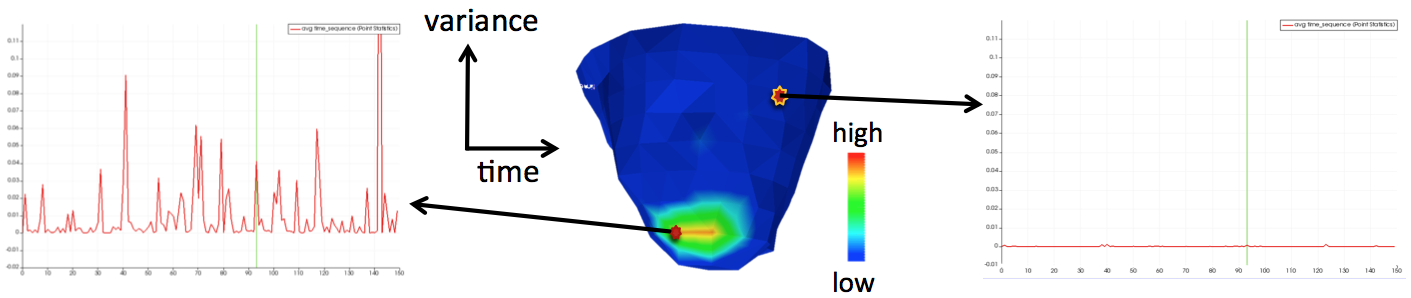}
\caption{ Centre shows spatial map of variance at various locations of the heart at a time. Two sides show temporal plots of the estimated variance at a site of scar tissue (left) and healthy tissue (right) }
\label{Var}
\end{figure}
We test our method in a set of 17 experiments on image-derived heart-torso models with the myocardial scar set according to the AHA 17-segment model of the LV. 
120-lead ECG data are simulated and corrupted with 20 dB Gaussian noise for inverse reconstruction. 
From the reconstructed action potential, 
activation time is extracted and scar is identified as the regions 
with an activation time larger than 
a threshold value. 
The accuracy of the detected scar is measured with two metrics: a) dice coefficient=2($S_1\cap S_2$)/($S_1 \cup S_2$), where $S_1$ and $S_2$ are reconstructed and true regions of scar, b) distance from the reconstructed center to actual center of scar.   

Fig. \ref{Modelcorrection} gives an example where a myocardial scar is located towards the anterior base of the LV. The \emph{a priori} EP model, however, assumes normal tissue property throughout the myocardium. As a result, model error exists in the predictions made by the EP model once the action potential propagates through the region of the myocardial scar. Compared to the method without model error correction, the proposed method is able to make significant improvements in two fronts: 1) the presented method corrects model errors significantly to drive the solution closer to the ground truth at each time instant, and 
2) the reconstructed solution is corrected toward the ground truth much earlier in time. 
As shown in Fig. \ref{Var}, 
higher variance of the wavefront is obtained 
at the region near myocardial scar, signifying possible errors in the model at those regions of the wavefront. This increase in prior variance of wavefront helps  make significant data-driven correction to the \emph{a priori} prediction. The temporal trace of the variance also shows that the estimated variance is  consistently  higher around the scar than that around the healthy tissue, suggesting low model errors in the latter region. 

Fig. \ref{ATandmetric}b lists the quantitative accuracy of the detected scar with and without model error estimation, where a statistically significant improvement of accuracy is obtained with model error estimation (unpaired-\emph{t} tests, p$<$0.005). 
Note that, Without error estimation, 
late activation due to scar is not captured in 50\% of the cases, 
and the metrics are calculated only from those where late activation is successfully reconstructed. 
Several examples of the reconstructed activation time 
are shown in Fig. \ref{ATandmetric}a, 
demonstrating the 
improvement of the delayed activation 
reconstructed at the region of myocardial scar 
by the presented method.

\begin{figure}[tb!]
\centering
\includegraphics[width=.8\textwidth]{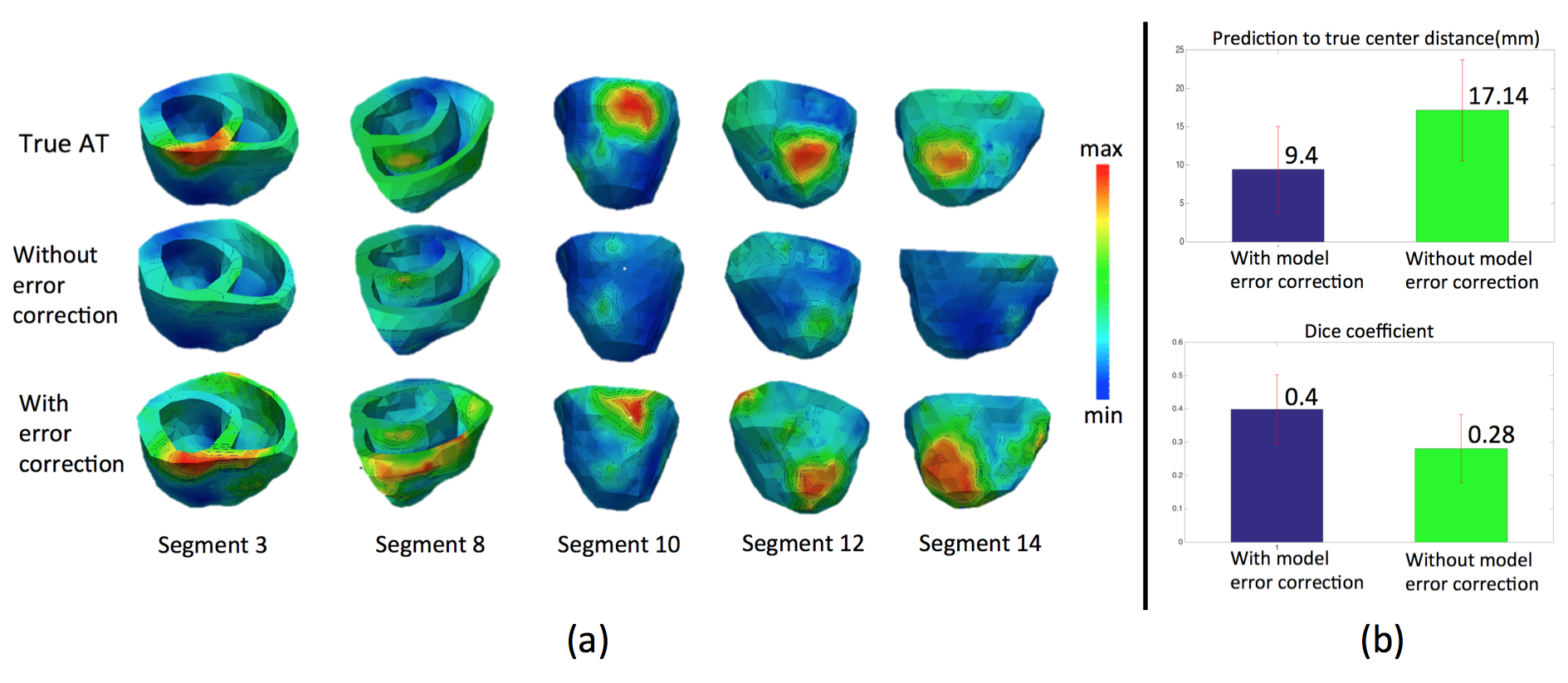}
\caption{a) Activation time maps. Regions with high values of activation time correspond to regions of scar. b) Comparison of the accuracy of scar detection with and without model error correction.  }
\label{ATandmetric}
\end{figure}

\begin{figure}[t!]
\centering
\includegraphics[width=.8\textwidth]{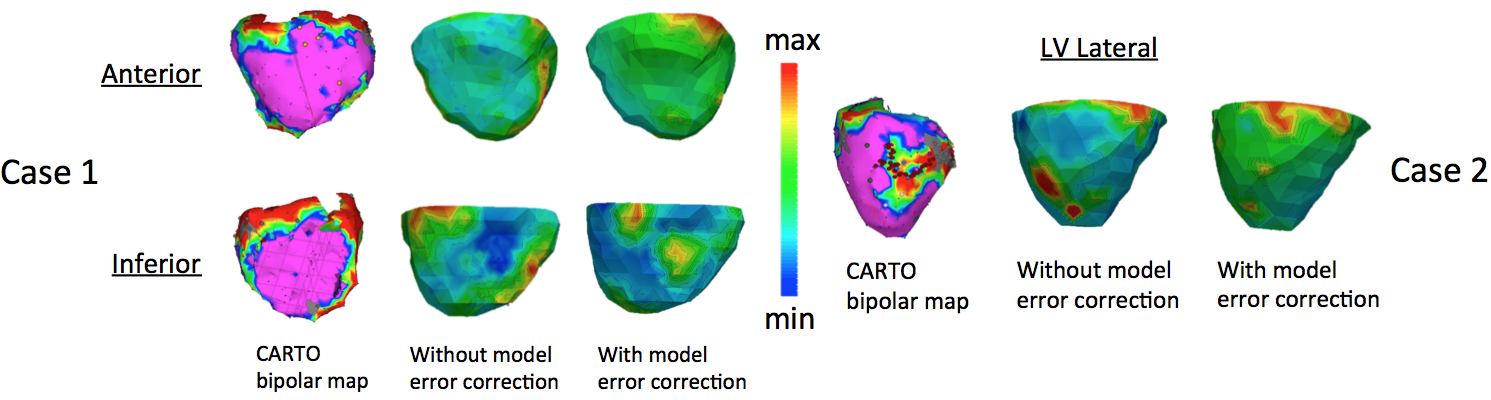}
\caption{Real data results of activation time maps versus \emph{in-vivo} bipolar voltage maps. In bipolar voltage maps, red regions correspond to scar and purple regions correspond to healthy tissue. In the reconstructed activation time maps, regions of high activation time correspond to scar.  }
\label{Realdata}
\end{figure}

\subsubsection{Real-data Experiment:}
Two real-data case studies are performed on patients who underwent catheter ablation due to scar-related ventricular arrhythmia. From 120-lead ECG data, action potential is reconstructed with and without model error estimation. Bipolar voltage data from \emph{in-vivo} catheter mapping is used as reference to evaluate the  detected scar regions. 
As shown in Fig.~\ref{Realdata}, in Case 1, the presented method  with model error correction is able to identify the area of scar at both  the anterior and inferior base. In contrast, reconstruction without error correction is only able to find scar in the inferior basal region. For Case 2, results from both methods are similar if one were to look at the region of high activation time (red region). However, at the mid-lateral region where voltage map shows the presence of scar, the activation time in the reconstruction with error correction is still high (green), while it is low (blue) for the reconstruction without error correction.
\section{Conclusion}
This paper presents a novel approach to model and estimate the error in \emph{a priori} EP models by exploiting its sparsity in the gradient domain. Experiments show promising results in the ability of the presented method to timely 
capture the error in the \emph{a priori} model 
at the correct spatial location. 
The next step would be to 
utilize the estimated variance 
to correct the model, 
\emph{e.g.}, 
by facilitating the estimation of 
its spatial parameters. 

\section*{Acknowledgement}
This work is supported by the National Science Foundation under CAREER Award ACI-1350374 and the National Institute of Heart, Lung, and Blood of the National Institutes of Health under Award R21Hl125998.

\bibliographystyle{splncs}
\bibliography{bibli1}

\end{document}